\documentclass[journal]{IEEEtran}  

\usepackage{dsfont}
\usepackage{graphicx}
\usepackage{epsfig}
\usepackage{subfig}
\usepackage{times}
\usepackage{amsmath}
\usepackage{amssymb}
\usepackage{cite}
\usepackage{color}
\usepackage{arydshln}
\usepackage[left, pagewise]{lineno}

\newtheorem{problem}{Problem}[section]

\newtheorem{lemma}{Lemma}[section]
\newtheorem{theorem}{Theorem}[section]
\newtheorem{corollary}{Corollary}[section]

\graphicspath{{figure/}}

\newcommand{\diag}{\operatornamewithlimits{diag}}

\title{Formation Control and Network Localization via Distributed Global Orientation Estimation in $3$-D}
\author{Byung-Hun Lee$^{\dag}$ and Hyo-Sung Ahn$^{\dag}$ 
\thanks{$^{\dag}$The authors are with School of Mechatronics,
Gwangju Institute of Science and Technology, Gwangju, Korea.
E-mail: {bhlee@gist.ac.kr}; {hyosung@gist.ac.kr.}}%
}

\begin{document}

\maketitle 
\thispagestyle{empty}
\pagestyle{empty}

\begin{abstract}
In this paper, we propose a novel distributed formation control strategy, which is based on the measurements of relative position of neighbors, with global orientation estimation in 3-dimensional space. Since agents do not share a common reference frame, orientations of the local reference frame are not aligned with each other. Under the orientation estimation law, a rotation matrix that identifies orientation of local frame with respect to a common frame is obtained by auxiliary variables. The proposed strategy includes a combination of global orientation estimation and formation control law. Since orientation of each agent is estimated in the global sense, formation control strategy ensures that the formation globally exponentially converges to the desired formation in 3-dimensional space. 
\end{abstract}


\section{Introduction}
Cooperative control of multiple autonomous agents has attracted a lot of attention due to its practical potential in various applications. Formation control of mobile autonomous agents is one of the most actively studied topics in distributed multi-agent systems. Depending on the sensing capability and controlled variables, various problems for the formation control have been studied in the literature \cite{Oh:Automatica2015}. 

Consensus-based control law for the desired formation shape is known as the displacement-based approaches. In this approach, agents measure the relative positions of their neighbors with respect to a common reference frame. Then, agents control the relative position for the desired formation. In the literature, it has been known that this method achieves global asymptotic convergence of the formation to the desired one \cite{Lin:TAC2005,Ren:IJRNC2007,Dimarogonas:Automatica2008,Oh:TAC2014}. Fundamental requirement for the displacement-based approach is that all agents share a common sense of orientation. Therefore, local reference frames need to be  aligned with each other. Oh and Ahn \cite{Oh:TAC2014} found the displacement-style formation control strategies based on orientation estimation, under the distance-based setup (i.e., local reference frames are not aligned with each other). In this approach, agents are allowed to align the orientation of their local coordinate systems by exchanging their relative bearing measurements. By controlling orientations to be aligned, the proposed formation control law allows agent to control relative positions simultaneously. 
Similar strategy using the consensus protocol is found in \cite{Montijano:ACC2014}. In the paper, control objective is to reach a desired formation with specified relative positions and orientations in the absences of a global reference frame. 
These strategies ensure asymptotic convergence of the agents to the desired formation under the assumption that interaction graph is uniformly connected and initial orientations belong to an open interval $[0,\pi)$\cite{Oh:TAC2014, Montijano:ACC2014}. In our previous work\cite{Lee:Automatica2016}, we proposed a novel formation control strategy via the global orientation estimation. The orientation of each agent is estimated based on the relative angle measurement and auxiliary variables assigned in the complex plane. Unlike the result of \cite{Oh:TAC2014,Montijano:ACC2014}, in this new approach, consensus property, which is based on the Laplacian matrix with complex adjacency matrix, is applied to allow the global convergence of the auxiliary variables to the desired points. In this way, we showed that the global convergence of formation to desired one for the multi-agent system having misaligned frames could be achieved.  

In this paper, we extend the result of \cite{Lee:Automatica2016} to 3-dimensional space. 
Even though we share the philosophy of the algorithm of \cite{Lee:Automatica2016}, there are intrinsic distinctions from \cite{Lee:Automatica2016} as follows. First, in 2-dimensional space, an unique auxiliary variable for each agent is required to estimate the orientation of the local frame. However, for the estimation of orientation in 3-dimensional system, we require more than two auxiliary variables for each agent, since at least two auxiliary variables are required for obtaining the rotation matrix in $\mathrm{SO}(3)$. Second, while the direction of the unique auxiliary variable can be considered as the orientation in 2-dimensional space, the orientation in 3-dimensional space can not be determined directly from the two different directions of auxiliary variables. Thus, we need to have additional processes compared to the case of 2-dimensional space for obtaining the estimated orientation. 
The novel method proposed in this paper uses the consensus protocol with auxiliary variables to estimate a rotation matrix in $\mathrm{SO}(3)$. This consensus protocol achieves global convergence of auxiliary variable, since configuration space of auxiliary variable is a vector space instead of nonlinear space like circle. We show that auxiliary variables construct the rotation matrix that identifies the estimated orientation. For the case of general $n$-dimensional space, the estimation method for unknown orientation is proposed in our previous work\cite{Lee:ICARCV2016}. In this paper, we modify the estimation method for the case of 3-dimensional space, and combine it with the formation control law for the desired formation shape. The estimated direction of the frame is used to compensate misaligned frames of each agent in the formation system. Consequently, the formation control law with the estimated coordinate frame achieves a global convergence of multi-agent system to desired shape in 3-dimensional space, under the distance-based setup. 
Third, we also show that the proposed strategy can be applied to the localization problem of sensor networks equipped with displacement and orientation measurement sensors. We design the position estimation law for the given displacement measurements. Then, the position estimation law is combined with the orientation estimation law. It thus shows that convergence property of localization algorithm is identical to the one of formation control algorithm.


The outline of this paper is summarized as follows. In Section II, preliminaries are described. In Section III, an idea for estimation law is described and stability of estimation law is analyzed. Convergence property of the formation control strategy with orientation estimation is then analyzed in Section IV. By using the proposed method, localization algorithm in networked systems is analyzed in Section V. Simulation results are provided in Section VI. Concluding remarks are stated in Section VII.

\section{Preliminaries}	\label{preliminaries}
In this paper, we use the following notation. Given $N$ column vectors $ x_1$, $x_2 \dots x_N \in \mathbb R^n$, $\bold x$ denotes the stacking of the vectors, i.e. $\bold x =( x_1,   x_2, \dots, x_N) $. The matrix $I_n$ denotes the $n$-dimensional identity matrix. Given two matrices $A$ and $B$, $A \otimes B$ denotes the Kronecker product of the matrices. The \emph{exterior product}, also known as \emph{wedge product}, of two vectors $u$ and $v$ is denoted by $u \wedge v$. 
\subsection{Consensus Property}
Given a set of interconnected systems, the interaction topology can be modeled by a directed weighted graph $\mathcal G = (\mathcal V, \mathcal E, \mathcal A)$, where $\mathcal V$ and $\mathcal E$ denote the set of vertices and edges, respectively, and $\mathcal A$ is a weighted adjacency matrix with nonnegative elements $a_{kl}$ which is assigned to the pair $(k,l)\in \mathcal E$. Let $\mathcal N_i$ denote neighbors of $i$. The Laplacian matrix $L = [l_{ij}]$ associated to $\mathcal G = (\mathcal V, \mathcal E, \mathcal A)$ is defined as $l_{ij} = \sum_{k\in \mathcal N_i} a_{ik}$ if $ i= j$ and $l_{ij}= -a_{ij}$ otherwise.  If there is at least one node having directed paths to any other nodes, $\mathcal G$ is said to have a rooted-out branch. In the following, we only consider the graph $\mathcal G$ with a rooted-out branch. The following result is borrowed from \cite{Moreau:CDC2004}.
\begin{theorem}
Every eigenvalue of $L$ has strictly positive real part except for a simple zero eigenvalue with the corresponding eigenvector $[1,1,\dots 1]^T$, if and only if the associated digraph has a rooted-out branch.
\label{thm_laplacian}
\end{theorem}

Let us consider $N$-agents whose interaction graph is given by $\mathcal G =(\mathcal V, \mathcal E, A)$. A classical consensus protocol in continuous-time is as follows \cite{Luca:Automatica2009}:
\begin{eqnarray}
\dot{ x}_i(t) = \sum_{j \in \mathcal N_i} a_{ij} ( x_j(t) -  x_i(t)) , ~ \forall i \in \mathcal V
\label{eq1}
\end{eqnarray}
where $ x_i \in \mathbb R^n$. Using the Laplacian matrix, (\ref{eq1}) can be equivalently expressed as 
\begin{eqnarray}
\dot{\bold x}(t) = -L\otimes I_n \bold{x}(t)
\label{eq2}
\end{eqnarray}
The following result is a typical one. 
\begin{theorem}\label{thm_consensus}
The equilibrium set $\mathcal E_x = \{ \bold x \in \mathbb R^{nN} :  x_i =  x_j,~ \forall i,j \in \mathcal V\}$ of (\ref{eq2}) is globally exponentially stable if and only if $\mathcal G$ has a rooted-out branch. Moreover, the state $\bold x(t)$ converges to a finite point in $\mathcal E_x$. 
\end{theorem}

\subsection{Orientation Alignment}
We assume that relative orientation and relative displacement measurements are only available in the formation control problem. The orientation of local reference frame can be controlled to be aligned with others. For the alignment of orientation, we can consider the consensus algorithm of (\ref{eq1}) based on the relative quantities of states i.e. $(x_j-x_i), ~ j\in \mathcal N_i$. In the case of $\mathrm{SO}(2)$, the orientation of $i$-th agent is denoted by $\theta_i \in \mathrm{S}^1$. Let us denote the displacement between $\theta_j$ and $\theta_i$ as $\theta_{ji}$. Note that $\theta_{ji} = \theta_{ji} - 2\pi $ and in general, $\theta_{ji} \neq \theta_j- \theta_i~.$ The consensus protocol in the unit circle space is as follows\cite{Oh:TAC2014, Montijano:ACC2014}:
\begin{eqnarray}\label{eq_consensus_unit_circle}
\dot \theta_i = \sum_{j\in \mathcal N_i} a_{ij} \theta_{ji} ,~~ i \in \mathcal V
\end{eqnarray}
In Euclidean space, the convex hull of a set of points in $n$-dimensions is invariant under the consensus algorithm of (\ref{eq1}). The permanent contraction of this convex hull allows to conclude that the agents end up at a consensus value\cite{Jadbabaie:TAC2003,Sepulchre:ARC2011}. In unit circle space, convergence property of the consensus algorithm based on the relative quantities of states is analyzed in similar way, while the convergence of states to a consensus value is not possible for all initial value of states. In other words, there is a subset such that convex hull of the set of points is invariant. In the following example, it will be shown that there is a domain which does not achieve the contraction of the convex hull under the consensus algorithm (\ref{eq_consensus_unit_circle}). 

\begin{itemize}
\item \emph{Example of anti-synchronization}: Let us suppose that twenty agents on the unit circle have all-to-all interaction topology shown in Fig. \ref{fig_anti-consensus}. Consensus protocol based on the displacement is designed as follows : 
\begin{eqnarray} \label{eq_anti-consensus}
\dot{\theta}_i  = \sum_{j\in \mathcal N_i} \theta_{ji}, ~~~ i\in \{1,2,\dots,20\}
\end{eqnarray}
The initial values of $\theta_i$ are assigned on out of range of $\pi$ as the pattern which is shown in Fig.\ref{fig_anti-consensus}. The simulation result of (\ref{eq_anti-consensus}) in Fig.\ref{fig_anti-consensus} shows that consensus is not guaranteed on unit circle space.  
\end{itemize}
\begin{figure}[!tb]
\begin{center}
\includegraphics[width=0.2\textwidth]{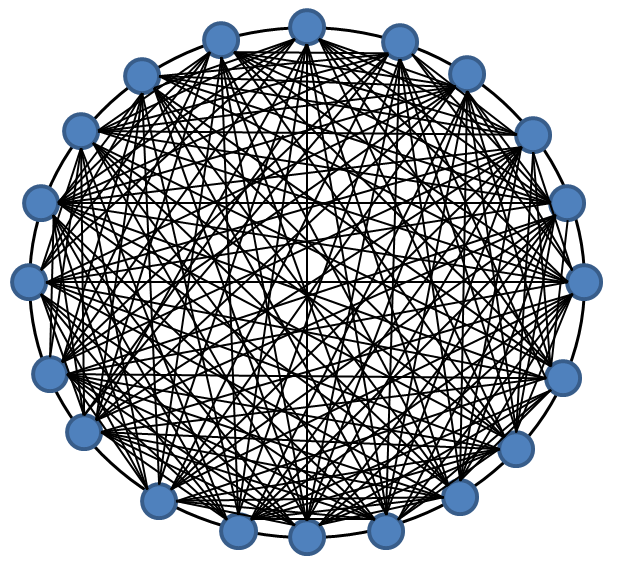}
\includegraphics[width=0.25\textwidth]{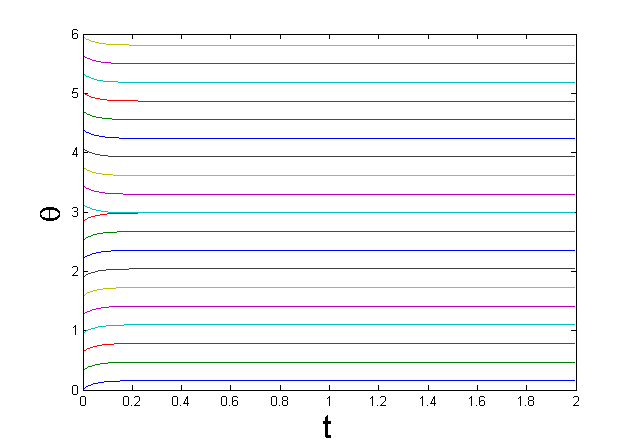}
\caption{Twenty particles on unit circle with all-to-all sensing graph (left figure) and the simulation result(right figure) of (\ref{eq_anti-consensus}) are demonstrated. }
\label{fig_anti-consensus}
\end{center}
\end{figure}
In the unit circle space, a sufficient condition for the permanent contraction of a convex hull is stated 
in the following theorem. 
\begin{theorem}[\cite{Oh:TAC2014, Dorfler:Automatica2014}]
Consider the coupled oscillator model (\ref{eq_consensus_unit_circle}) with a connected graph $\mathcal G(\mathcal V,\mathcal E, \mathcal A)$. Suppose that a convex hull of all initial values is within an open semi circle. Then, this convex hull is positively invariant, and each trajectory originating in the convex hull achieves an exponential synchronization. 
\end{theorem}

The fundamental difference between a consensus protocol (\ref{eq1}) in Euclidean space and another protocol (\ref{eq_consensus_unit_circle}) in unit circle space is the \emph{non-convex\footnote{\emph{Non-convexity} is the opposite of \emph{convexity}. Convexity property implies that the future, updated value of any agent in the network is a convex combination of its current value as well as the current values of its neighbors. It is thus clear that $\mathrm{max}\{x_1,x_2,\dots ,x_N\}$ is a non-increasing function of time\cite{Moreau:TAC2005}.}
  nature of configuration spaces} like circle or sphere. For this reason, a global convergence analysis of the consensus algorithm in nonlinear space is quite intractable and at least very dependent on the communication graph\cite{Sepulchre:ARC2011}.

In this paper, instead of controlling the orientation to be aligned, we estimate the orientation by using auxiliary variables defined on the higher dimensional vector spaces rather than unit circle space. Although consensus algorithm plays still important role in the proposed method, unlike the alignment law (\ref{eq_consensus_unit_circle}), the proposed method achieves global convergence of the estimated variables. This is described in the next section. 


\subsection{Problem Statement}
Consider $N$ single-integrator modeled mobile agents in the space:
$\dot { p}_i = u_i ~, ~i = 1,\dots N,$ 
where $ {p}_i\in \mathbb R^3$  and $u_i \in \mathbb R^3$ denote the position and control input, respectively, of agent $i$ with respect to the global reference frame, which is denoted by $^g \sum$. Following the notions of \cite{Oh:TAC2013}, we assume that agent $i$ maintains its own local reference frame with the origin at $ p_i$. An orientation of $i$-th local reference frame with respect to the global reference frame is identified by a proper orthogonal matrix $C_i \in \mathrm{SO}(3)$ as illustrated in Fig. \ref{fig_orien}. By adopting a notation in which superscripts are used to denote reference frames, the dynamics of agents are described as
$\dot {{p}}_i^i =  u_i^i, ~ i= 1,\dots N,$ 
where ${{p}}_i^i\in \mathbb R^3$ and $ u_i^i \in \mathbb R^3$ denote the position and control input, respectively, of agent $i$ with respect to $^i \sum$. 
For a weighted directed graph $\mathcal G = (\mathcal V, \mathcal E, \mathcal A)$, agent $i$ measures the relative positions of its neighboring agents with respect to $^i\sum$ as: 
\begin{eqnarray}
p^i_{ji} := p_j^i -  p_i^i, ~ j\in \mathcal N_i, i\in \mathcal V. 
\label{eq_relative_pos}
\end{eqnarray}
We note that $ p_j-  p_i =  C_i^{-1}(p_j^i-  p_i^i)$, because the local reference frames are oriented as much as the rotation matrix $C_i$ from the global reference frame. 
\begin{figure}[!tb]
\begin{center}
\includegraphics[width=0.3\textwidth]{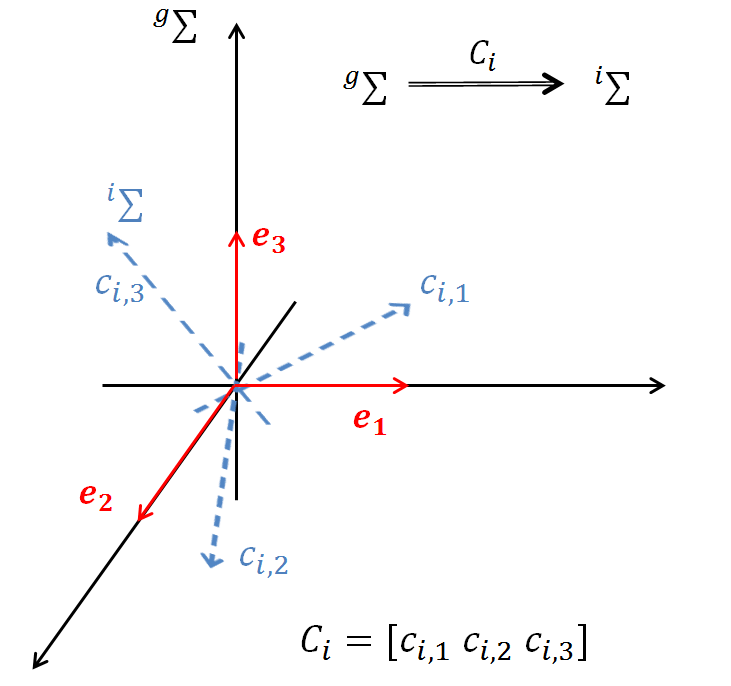}
\caption{The local coordinate frame of $i$-th agent is rotated from the global reference frame $\left(^g\sum \right)$ by $C_i$. $\{e_{i,1},e_{i,2},e_{i,3}\}$ is the set of orthonormal bases of $i$-th local frame $^i\sum$.}
\label{fig_orien}
\end{center}
\end{figure}

In this paper, we attempt to estimate $C_i$ by using the relative orientation measurement. Then, the estimated orientation is used to calibrate $C_i$ for the formation control. 
Let the desired formation $\bold p^* = ( p_1^*, \dots, p_N^*) \in \mathbb R^{3N}$ be given. The formation problem is then stated as follows:
\begin{problem} \label{prob_1}
Consider $N$-agents modeled by (\ref{eq1}). Suppose $\mathcal G = (\mathcal V, \mathcal E, \mathcal A)$ is the interaction graph of the agents. For a common reference frame $^a\sum$, design a control law such that $ p_j^a -  p_i^a \rightarrow  p_j^* -  p_i^*$ as $t\rightarrow \infty$ for all $i,j\in \mathcal V$ based on the measurements such as the relative displacement and relative orientation, which are measured in local coordinate frames of each agent.
\end{problem}

We next formulate the localization problem. Under the same assumption of the \emph{Problem} \ref{prob_1}, we desire to estimate the position of each agent. We denote $\hat p_i$ as the estimated position of $p_i$ with respect to the global reference frame. Then, the localization problem is stated in follows:
\begin{problem}\label{prob_2}
Consider $N$-agents modeled by (\ref{eq1}). Suppose $\mathcal G = (\mathcal V, \mathcal E, \mathcal A)$ is the interaction graph of the agents. For each agent $i\in \mathcal V$, design an estimation law of $\hat p_i$ such that $\| \hat p_j - \hat p_i \| \rightarrow \| p_j - p_i \|$ as $t\rightarrow \infty $ for all $j \in \mathcal V$ based on the measurements such as the relative displacement and relative orientation, which are measured in local coordinate frames of each agent.
\end{problem}

The following section proposes a method for the orientation estimation by using only relative orientation measurement. We describe the estimation method in general $n$-dimensional space instead of 3-dimensional space. 

\section{Novel Method for Orientation Estimation}
Orientation or attitude is often represented using three or four parameters in 3-dimensional space. Three or four-parameter representations include the unit-quaternions, axis-angle and Euler angles. While these parameters cannot represent orientation globally nor uniquely, the rotation matrix in $\mathrm{SO}(n)$ is unique and globally defined\cite{Chaturvedi:CSM2011}. 

The orientation in the $n$-dimensional space is defined by an $n\times n$ proper orthogonal matrix $C \in \mathrm{SO}(n)$ containing, as column vectors, the unit vectors $\bold c_i (i=1,2,\dots n)$ identifying the orientation of the local frame with respect to the reference frame. A matrix $C$ has properties such as $C^T C = I$ and $\mathrm{det}(C)= +1$. When the common reference frame is the Cartesian coordinate system with standard basis $\bold e_i $$(i=1,\dots,n)$, the orientation of local frame is represented as $\bold c_i =C \bold e_i$ which is the $i$-column vector of $C$. 
\subsection{Orientation estimation in $\mathrm{SO}(n)$}
Consider $N$ agents whose interaction graph is $\mathcal G$ that has a spanning tree in $n$-dimensional space. Each agent does not share a common reference frame. It follows that the $i$-th local reference frame is rotated from the global reference frame with the amount of $C_i\in \mathrm{SO}(n)$. Let us suppose that $C_{ji}$ denotes the orientation of $j$-th local reference frame with respect to the $i$-th local reference frame. {It provides $i$-th agent view to the $j$-th local reference frame with respect to the $i$-th local reference frame.}
If the global reference frame has the standard basis, each basis of $i$-th local frame with respect to the global reference frame is the column vector of $C_i$. It is well-known that $C_{ji}$ is defined by the multiplication of two rotation matrices as follows:
\begin{eqnarray}
C_{ji} = C_j C_i^{-1} = C_j C_i^T
\label{relative_orien}
\end{eqnarray}
We assume that the $i$-th agent measures the orientation of $j$-th agent with respect to the $i$-th local reference frame.
Therefore, $C_{ji}$ is obtained by the measurement. Let $\hat C_i \in \mathrm{SO}(n), \forall i\in \mathcal V$ be an estimated orientation which is the proper orthogonal matrix. Orientation estimation problem can be stated as follows:
\begin{problem}\label{problem1}
Consider $N$-agents whose interaction graph is $\mathcal G$ in $n$-dimensional space. For the common orientation which is identified by $C^* \in \mathrm{SO}(n)$, design an estimation law such that $C_i(t)^{T} \hat{ C}_i \rightarrow C^*$ as $t\rightarrow \infty$ for all $i \in \mathcal V$ based on the relative orientation (\ref{relative_orien}). 
\end{problem}

Let us denote $B_i\in \mathrm{SO}(n)$ which is defined as 
\begin{eqnarray}
B_i := C_i X ~,~X\in \mathrm{SO}(n), \forall i\in \mathcal V.
\end{eqnarray}
According to (\ref{relative_orien}), it follows that 
\begin{eqnarray}
B_j = C_{ji} B_i~, ~ \forall (i,j)\in \mathcal E.
\label{relation}
\end{eqnarray}
Under the assumption of \emph{Problem \ref{problem1}}, finding the steady state solution of $\hat C_i$ is equivalent to finding $B_i$ satisfying the equality (\ref{relation}). 
Then, \emph{Problem \ref{problem1}} can be restated as follows:
\begin{problem}\label{problem2}
Under the assumption of \emph{Problem \ref{problem1}}, design an evolutionary algorithm for $B_i:[0,\infty) \rightarrow \mathbb R^{n\times n}$ such that 
\begin{itemize}
\item $B_i(t) \in \mathrm{SO}(n) ~ ,~ \forall t\in [0,\infty)$
\item $B_j(t) \rightarrow C_{ji} B_i(t)$ as $t \rightarrow \infty$ , $\forall (i,j)\in \mathcal E$
\end{itemize}
\end{problem}

Let us suppose that $n-1$ auxiliary variables are generated at each agent. Each auxiliary variable at $i$-th agent is denoted by $z_{i,k} \in \mathbb R^{n}, k\in \{1,2,\dots, n-1\}$. Let $B_i$ be represented as $B_i=[b_{i,1}, b_{i,2},\dots,b_{i,n}]$ where $b_{i,k}\in \mathbb R^n, k\in\{1,\dots n\}$, is a $k$-th column vector of $B_i$. To generate orthonormal column vectors of $B_i$ from $z_{i,k}$ , we can use Gram-Schmidt process with any independent vectors $z_{i,k} \forall k\in \{1,\dots n-1\}$ as follows:
\begin{eqnarray}
\begin{array}{ll}
v_{i,1}  := z_{i,1}  &,b_{i,1} := \frac{v_{i,1}}{ \|v_{i,1}\|}  \\
v_{i,2} := z_{i,2} - \langle z_{i,2}, b_{i,1}\rangle b_{i,1} &, b_{i,2} := \frac{v_{i,2}}{ \|v_{i,2}\|}  \\
\quad \vdots &  ~~~~~~~\vdots\\
v_{i,m} := z_{i,m} - \sum_{k=1}^{m-1}\langle z_{i,m}, b_{i,k}\rangle b_{i,k}  &,  b_{i,m} :=  \frac{v_{i,m}}{ \|v_{i,m}\|}
\end{array}\label{gram-schmidt}
\end{eqnarray}
where $m = n-1$ and $\langle \cdot,\cdot \rangle$ denotes the operator of inner product. The rest of the procedure include determining $b_{i,n}$ such that the determinant of $B_i$ is equivalent to $1$. We consider $b_{i,n}$ as a \emph{pseudovector}\footnote{Pseudovector is a vector-like object which is invariant under inversion. In physics a number of quantities behave as pseudovector including magnetic field and angular velocity. In 3-dimensional vector space, the pseudovector $p$ is associated with the cross product of two vectors $a$ and $b$ which is equivalent to 3-dimensional bivectors : $p = a \times b$. }
 for making det($B_i$) to be one. Pseudovector is a quantity that transforms like a vector under a proper rotation. In $n$-dimensional vector space, a pseudovector can be derived from the element of the ($n-1$)-th exterior powers, denoted $\wedge^{(n-1)}(\mathbb R^n)$. It follows that
\begin{eqnarray} \label{exterior_power}
b_{i,1} \wedge  b_{i,2} \wedge \cdots \wedge b_{i,n-1} . 
\end{eqnarray}
Since every vector can be written as a linear combination of basis vectors, $\wedge^{(n-1)}(\mathbb R^n)$ can be expanded to a linear combination of exterior products of those basis vectors. The pseudovector can be obtained from coefficients of (\ref{exterior_power}). Since we obtain $b_{i,n}$ from the formalization of pseudovector, the matrix $B_i$ is a proper rotation.

This is one way to calculate the quantity $b_{i,n}$ in  higher dimensional space, while there is no natural identification of $\wedge^{(n-1)}(\mathbb R^n)$ with $\mathbb R^n$. Another convention for obtaining $b_{i,n}$ with proper rotation $B_i$ is stated in \cite{Lee:ICARCV2016}. 

Now, we consider a second goal of the \emph{Problem \ref{problem2}}. We propose a single integrator model for dynamics of auxiliary variables as follows : 
\begin{eqnarray} \label{z_model}
\dot z_{i,k} = u_{i,k} \quad, \forall i\in \mathcal V, \forall k\in \{1,\dots, n-1\}.
\end{eqnarray}
Under the assumption of the \emph{Problem \ref{problem1}}, a control law for (\ref{z_model}) is designed as follows for all $ i\in \mathcal V$ and $k\in \{1,\dots, n-1\}$: 
\begin{eqnarray}
u_{i,k} = \sum_{j\in \mathcal N_i} a_{ij}(C_{ji}^{-1} z_{j,k} - z_{i,k}) 
\label{dynamics}
\end{eqnarray}
where $a_{ij}> 0$. From the definition of $C_{ji}$, we have $C_{ji}^{-1} = C_{ji}^T = C_{ij}$. Let $\bold z_k$ be a stacked vector form defined by $\bold z_k := ( z_{1,k},z_{2,k},\dots, z_{N,k})$. Here, (\ref{z_model})-(\ref{dynamics}) can be written in terms of $\bold z_k$ as follows:
\begin{eqnarray} \label{vec_z_model}
\dot {\bold z}_k = H \bold z_k
\end{eqnarray}
where $H\in \mathbb R^{nN\times nN}$ is a block matrix which is defined as 
\begin{eqnarray}
H := \left[\begin{array}{cccc}
H_{11} & H_{12} & \cdots & H_{1N} \\
H_{21} & H_{22} & \cdots & H_{2N} \\
\vdots  & \vdots  & \ddots & \vdots \\
H_{N1} & H_{N2} & \cdots & H_{NN} \\
\end{array} \right]
\end{eqnarray}
Each partition $H_{ij} \in \mathbb R^{n\times n}$ is written as follows:
\begin{eqnarray}
H_{ij} = \left\{\begin{array}{cl}
a_{ij} C_{ij} &, j\in \mathcal N_i\\
-\sum_{j\in \mathcal N_i}a_{ij} I_n &, i=j\\
\bold 0_n &, j\notin \mathcal N_i 
\end{array} \right.
\end{eqnarray}
 where $\bold 0_n \in \mathbb R^{n\times n}$ is a matrix having only zero entries. 
 \subsection{Analysis of the Stability}
 Let us consider the eigenvalue of $H$. By using a nonsingular matrix $D$, similarity transformation can be achieved as follows : $\bar H := D^{-1} H D$. Suppose that $D$ is a block diagonal matrix defined by $D := \diag(C_1 , C_2,\dots ,C_N)$ where $C_i \in \mathbb R^{n\times n}$, $i=\{1,\dots N\} $. Suppose that $\bar H_{ij}\in \mathbb R^{n\times n}$ denotes a partition of $i$-th row and $j$-th column in the block matrix $\bar H$. $\bar H_{ij}$ is written as follows :
\begin{eqnarray}
\bar H_{ij} = C_i^T H_{ij} C_j = \left\{ \begin{array}{cl}
a_{ij}I_n&  ,j \in \mathcal N_i\\
-\sum_{j\in \mathcal N_i} a_{ij} I_n& , i = j\\
\bold 0_n & , j\notin \mathcal N_i.
\end{array}\right.
\label{bar_H}
\end{eqnarray}
From (\ref{bar_H}), $\bar H$ is rewritten as $\bar H= -L_H \otimes I_n$ where $L_H\in \mathbb R^{N\times N}$  has zero row sum with dominant diagonal entries. From the \emph{Theorem} \ref{thm_laplacian}, all eigenvalues of $L_H$ have strictly positive real part except for a simple zero eigenvalue with corresponding eigenvector $\xi = (1,1,\dots, 1) \in \mathbb R^{N}$. 
 
 Let us consider a coordinate transformation as $\bold z_k = D \bold {q}_k$. Then (\ref{vec_z_model}) is rewritten as follows:
\begin{eqnarray}
 \bold {\dot {q}}_k(t) = -L_H\otimes I_{n-1} {\bold {  q}}_k(t),  ~~\forall k \in \{1,\dots, n-1\} 
 \label{laplace}
\end{eqnarray}
From the \emph{Theorem} \ref{thm_consensus}, it is clear that $\bold{ q}_k$ converges to the equilibrium set which is indicated as $\mathcal E_q = \{ \bold x= (x_1,x_2,\dots,  x_N) \in \mathbb R^{nN} : x_1 = x_2 =\cdots = x_N\} $. For the estimation of orientation, we have to avoid the convergence of variables $\bold{q}_k$ to zero. Therefore, the desired equilibrium set is defined by $\mathcal S_q := \mathcal E_q \setminus \{0\}$. For a square matrix $A$, define the column space of matrix $A$ as $\bold{C}(A)$. Then, we have the relationship $\bold{C}(A) = \rm{null}$$(A^T)^\perp$, where $\mathrm{null}$$(A^T)^\perp$ denotes the orthogonal space of null space of $A^T$. Based on the consensus property mentioned in the previous section, the following lemma provides conditions for the convergence of $\bold {q}_k $ to $\mathcal S_q$. 
\begin{lemma}
Suppose that $\mathcal G$ has a rooted-out branch. For the dynamics (\ref{laplace}), there exists a finite point $\bold{ q}^\infty_k \in \mathcal S_q$ for each $k\in \{ 1,\dots n-1\}$ such that $\bold{ q}_k$ exponentially converges if and only if an initial value $\bold{ q}_k(t_0), \forall k\in \{1,2,\dots,n-1\}$ is not in $\bold C(L_H \otimes I_{n})$. 
\label{thm_conv} 
\end{lemma}
\begin{proof}
From the \emph{Theorem} \ref{thm_consensus}, $\bold{q}_k, \forall k\in \{1,2,\dots n-1\}$ globally exponentially converges to $\mathcal E_q$. Then, for each $k\in \{1,2,\dots n-1\}$, there exists a finite point $\bold{ q}^\infty_k \in \mathcal E_q$ and constants $\lambda_q,\eta_q>0$ such that
\begin{eqnarray}
\| \bold{ q}_k(t) - \bold{ q}^\infty_k\| \leq \eta_q \mathrm{e}^{-\lambda_q (t-t_0) }\| \bold{ q}_k(t_0) - \bold{ q}^\infty_k\|.
\end{eqnarray}
Now, we consider a steady state solution $\bold { q}^\infty_k$. A solution of $\bold {q}_k$ is written as follows: 
\begin{eqnarray}
\bold { q}_k(t) =  \mathrm{e}^{-L_H\otimes I_{n} (t-t_0)} \bold{ q}(t_0). \label{q_conv}
\end{eqnarray}
The Jordan form of $L_H$ is obtained by the similarity transformation as follows: $G^{-1}L_H G= J$. Let $G$ and $G^{-1}$ be represented as $G = [g_1~\cdots~g_N]$ and $G^{-1} = [w_1^T ~ \cdots ~ w_N^T]^T$ respectively. We supposed that $g_N$ is the right eigenvector of the zero eigenvalue. Then, $w_N$ is the left eigenvector of the zero eigenvalue. Since every nonzero eigenvalue of $L_H$ has negative real part, the state transition matrix $e^{J\otimes I_{n}(t-t_0)}$ has the following form as $t\rightarrow \infty$.
\begin{eqnarray}
\mathrm{e}^{J\otimes I_{n}(t-t_0)} \rightarrow \left[\begin{array}{cccc}
0 & & &\\
 &0 & &\\
 & & \ddots & \\
 & & &1 
\end{array} \right] \otimes I_{n}
\label{stm}
\end{eqnarray}
From (\ref{stm}), the steady state solution of $\bold{{q}}_k$ is as follows:
\begin{align}
&\lim_{t\rightarrow \infty} \bold{ q}_k (t)  \nonumber\\
&\quad = \lim_{t\rightarrow \infty} (G \otimes I_{n}) \mathrm{e}^{J\otimes I_n(t-t_0)}(G^{-1} \otimes I_{n})  \bold{ q}_k (t_0) \nonumber\\
&\quad = (g_N\otimes I_{n})( w_N \otimes I_{n})\bold{ q}_k (t_0)\nonumber\\
&\quad = (g_N w_N \otimes I_{n})\bold{q}_k (t_0) \label{sss_q}
\end{align}
This implies that $\bold{ q}_k (t)$ converges to zero if and only if $\bold{q}_k (t_0)$ is perpendicular to the $w_N \otimes I_{n}$ which is the left eigenvector of the zero eigenvalue. Then, it follows that $(w_N \otimes I_{n})\bold{q}_k(t_0)=0$ if and only if $\bold{q}_k(t_0)$ is in $\rm{null}$$(L_H^T\otimes I_{n})^\perp$. Since $\rm{null}$$(L_H^T\otimes I_{n})^\perp = \bold{C}(L_H\otimes I_{n})$, it completes the proof. 
\end{proof}

The result of \emph{Lemma \ref{thm_conv}} implies $\lim_{t\rightarrow \infty} \bold{z}_k(t) = D \bold q_k^\infty$ for each $k$. Let $B_i \in \mathrm{SO}(n)$ be derived from the procedure including Gram-Schmidt process of $z_{i,k}, \forall k\in \{1,\dots n-1 \}$ and the calculation of $b_{i,n}$. Consequently, there exists $P^\infty_i \in \mathrm{SO}(n)$ to which $B_i$ converges. We can show that $B_i(t) \rightarrow C_{ij} B_j(t), (i,j) \in \mathcal E$ as $t \rightarrow \infty$ from the following result. 
\begin{theorem}\label{thm_P}
Let $C_i \in \mathrm{SO}(n)$ denote a rotation matrix which identifies the orientation of the $i$-th agent. Suppose that $n\times n$ matrix $B_i$ for the $i$-th agent is derived from the procedures (\ref{gram-schmidt}), (\ref{exterior_power}) and (\ref{vec_z_model}). Then, there exists a common matrix $X\in \mathrm{SO}(n)$ such that $B_i$ converges to $C_i X$ as $t\rightarrow \infty$ for all $i\in\{1,\dots N\}$. 
\end{theorem}
\begin{proof}
By using the coordinate transformation, we define $q_{i,k}\in \mathbb R^n$ for all $i\in\{1,\dots N\}$ and $k\in \{1,\dots, n-1\}$ such that $z_{i,k} = C_i q_{i,k}$. Let us replace the value of $q_{i,k}$ with $z_{i,k}$ in the procedure (\ref{gram-schmidt}). From the definition of $b_{i,k}$ for all $k$ in (\ref{gram-schmidt}), $v_{i,k}$ is rewritten as follows :
\begin{eqnarray}\label{seq_v}
\begin{array}{rl}
v_{i,1} &= C_i q_{i,1} := C_i x_{i,1}\\
v_{i,2} &= C_i   q_{i,2} - \langle  C_i q_{i,2},  \frac{v_{i,1}}{\|v_{i,1} \|} \rangle \frac{v_{i,1}}{\|v_{i,1} \|} \\
&:= C_i x_{i,2}    \\\
\quad \vdots  \\ 
v_{i,n-1} &= C_i   q_{i,n-1} - \sum_{k=1}^{n-2} \langle C_i q_{i,n-1},  \frac{v_{i,k}}{\|v_{i,k} \|} \rangle \frac{v_{i,k}}{\|v_{i,k}\|}   \\
& := C_i x_{i,n-1}
\end{array}
\end{eqnarray}
Using this to replace $v_{i,m}$ by $x_{i,m}, \forall m \in \{1,\dots n \}$,  $x_{i,m}$ can be written as follows: 
\begin{eqnarray}
\begin{array}{rl}
x_{i,1} &= q_{i,1} \\
x_{i,2} &=    q_{i,2} - \langle  q_{i,2},  \frac{x_{i,1}}{\|x_{i,1} \|} \rangle \frac{x_{i,1}}{\|x_{i,1} \|} \\
\quad \vdots  \\ 
x_{i,n-1} &=   q_{i,n-1} - \sum_{k=1}^{n-2} \langle q_{i,n-1},  \frac{x_{i,k}}{\|x_{i,k} \|} \rangle \frac{x_{i,k}}{\|x_{i,k}\|}   
\end{array}\label{seq_x}
\end{eqnarray}
This is a Gram-Schmidt procedure with respect to the $q_{i,k}, \forall k\in \{1,\dots n-1 \}$.  
   From the result of \emph{Lemma \ref{thm_conv}}, for each $k$, $q_{i,k}$ converges to a finite point $q_k^\infty \in \mathbb R^n $ as $t\rightarrow \infty$ for all $i\in \mathcal V$. This implies that there exists a finite vector $x_k^\infty \in \mathbb R^n$ such that $\{x_{i,1}, x_{i,2}\cdots x_{i,n-1} \} \rightarrow \{x_1^\infty, x_2^\infty, \cdots x_{n-1}^\infty \}$ as $t\rightarrow \infty$ for all $i$. It is clear that $b_{i,k}$ and $x_{i,k}$ have the following relationship: $b_{i,k} = C_i \frac{x_{i,k}}{\| x_{i,k}\|}, \forall k\in\{1,2\dots n-1\}$. We also derive $b_{i,n}$ from (\ref{exterior_power}).
   Therefore, in the steady state of $x_{i,k}$, we have the following solution of $B_i(t)$:
\begin{eqnarray}
\lim_{t\rightarrow \infty} B_i(t) = C_i \left[\frac{x_1^\infty }{ \| x_1^\infty\| } \frac{x_2^\infty }{ \| x_2^\infty \|} \cdots  \frac{x_n^\infty }{\|x_n^\infty \|} \right] , \quad \forall i\in \mathcal V
\end{eqnarray}
It completes the proof. 
\end{proof}

Since $B_i(t) \rightarrow C_i X$ as $t\rightarrow \infty$, it is clear that $B_i(t) \rightarrow C_{ij} B_j(t)$ as $t\rightarrow \infty$. From the result of \emph{Theorem \ref{thm_P}}, we now see that the proposed algorithm provides the solution of the \emph{Problem \ref{problem2}}. Since we transformed the \emph{Problem \ref{problem1}} to \emph{Problem \ref{problem2}}, $B_i$ is the estimated orientation. 

In the following section, we combine this algorithm with the formation control. For a practical reason, we consider the formation control in 3-dimensional spaces instead of the general $n$-dimensional spaces. 

\section{Formation Control In 3-D Space}
In 3-dimensional space, we have to obtain the matrix $B_i\in \mathrm{SO}(3)$ satisfying conditions stated in \emph{Problem \ref{problem2}}. We have two auxiliary variables denoted by $z_{i,1} , z_{i,2}$ for the $i$-th agent. By using the proposed estimation law (\ref{vec_z_model}) and Gram-schmidt process, we can obtain $b_{i,1}$ and $b_{i,2}$. Since the pseudovector is directly calculated by using a cross product of two vectors in 3-dimensional space, the quantity $b_{i,3}$ is as follows :
\begin{eqnarray}
b_{i,3} = b_{i,1} \times b_{i,2},~~~ \forall i\in \mathcal V. 
\end{eqnarray}
Note that $\hat C_i(t) = B_i(t) \in \mathrm{SO}(3)$. 
We now propose a control law for the formation system under the assumption of \emph{Problem} \ref{prob_1}. Consider the following control law:
\begin{eqnarray}
u_i^i = k_u\sum_{j\in \mathcal N_i} l_{ij} ((p^i_j - p_i^i) - \hat C_i(p_j^* - p_i^*)) \label{cont_law}
\end{eqnarray}
where $k_u>0$ and $l_{ij} >0$. Since $u_i = C_i^{-1} u_i^i$, the position dynamics of $i$-th agent with respect to the global reference frame is written as follows:
\begin{eqnarray}
\dot p_i &=& u_i = C_i^{-1}u_i^i \\
&=& k_u \sum_{j\in \mathcal N_i} l_{ij} C_i^{-1}\left((p^i_j - p_i^i) - \hat C_i(p_j^* - p_i^*)\right)\\
&=& k_u \sum_{j\in \mathcal N_i} l_{ij} \left((p_j - p_i) - C_i^{-1}\hat C_i(p_j^* - p_i^*)\right) 
\end{eqnarray}
From the \emph{Theorem} \ref{thm_P}, there exists $C^*\in \mathbb R^{3\times 3}$ such that $C_i^{-1}\hat C_i(t)$ converges to $C^*$ as $t\rightarrow \infty$. By defining ${e}_{i} := p_i- C^* p_i^* $, we obtain the error dynamics as follows: 
\begin{eqnarray}
\dot{{e}}_{i} &=& k_u \sum_{j\in \mathcal N_i} l_{ij}  ({e}_{j} - {e}_{i})\nonumber\\
&&  + \underbrace{k_u \sum_{j\in \mathcal N_i} l_{ij}(C^*- C^{-1}_i \hat C_i)(p_j^* - p_i^*)}_{w_i:=}  \label{err_dynamic}
\end{eqnarray}
Then, (\ref{err_dynamic}) can be rewritten in the vector form as follows:
\begin{eqnarray}\label{cont_err}
\bold{\dot e} = -k_u L\otimes I_3 \bold{e} + \bold w
\end{eqnarray}
where $\bold e = ({e}_1,\dots ,{e}_N)$ and $\bold w = (w_1, \dots , w_N)$. From (\ref{q_conv}) and the result of the \emph{Theorem} \ref{thm_P}, it is clear that $C_i^{-1}\hat C_i$ converges to $C^*$ with an exponential rate. It follows that there exist $k_w>0$ and $\lambda_w >0$ such that 
\begin{eqnarray}
\|\bold w(t)\|  \leq k_w \mathrm{e}^{-\lambda_w(t-t_0)}\| \bold{w}(t_0)\| \label{w_conv}
\end{eqnarray}
Now, we consider the solution of $\bold e(t)$ given by
\begin{eqnarray}
\bold e(t) = \mathrm{e}^{-L\otimes I_3 (t-t_0)} \bold e(t_0) + \int^t_{t_0} \mathrm{e}^{-L\otimes I_3 (t-\tau)} \bold w(\tau) d\tau. 
\end{eqnarray}
By using the result in \emph{Theorem} 3.1 of \cite{Oh:TAC2014}, the following result is obtained. 
\begin{theorem}
Under the estimation law (\ref{vec_z_model}) and the formation control law (\ref{cont_law}), there exists a finite point $ p_\infty$ such that $p_i(t), \forall i\in \mathcal V$ globally exponentially converges to $C^* p_i^*+p_\infty $, if $\mathcal G$ has a rooted-out branch, and $\bold{\hat q}_k(t_0), \forall k\in\{1,2\} $is not in $\rm{C}$$(L_H \otimes I_3)$. \label{thm_formation}
\end{theorem}
\begin{proof}
Let us define $\mathcal E_e$ as 
\begin{eqnarray}
\mathcal E_e := \{ (x_1 ,\dots,x_N)\in \mathbb R^{3N} : x_i= x_j ,\forall i,j\in \mathcal V\}
\end{eqnarray}
There exists $\bold e^*\in \mathcal E_e$, $k_l>0$ and $\lambda_l >0$ such that 
\begin{eqnarray}
\| \bold e(t) - \bold e^*\| &\leq & \left\| \mathrm{e}^{-L\otimes I_3 (t-t_0)} \bold e(t_0) -\bold e^* \right\| \\
&& + \left\|\int^t_{t_0} \mathrm{e}^{-L\otimes I_3 (t-\tau)} \bold w(\tau) d\tau \right\| \\
& \leq & k_l \mathrm{e}^{-\lambda_l (t-t_0)} \| \bold e (t_0) - \bold e^*\| \\
&&+ \int^t_{t_0} \left\|\mathrm{e}^{-L\otimes I_3 (t-\tau)}\right\|  \left\|\bold w(\tau)  \right\| d\tau \label{err_bound}
\end{eqnarray}
According to the definition of induced norm, $\left\|\mathrm{e}^{-L\otimes I_3 (t-\tau)}\right\|$ is written as 
\begin{eqnarray}
\left\|\mathrm{e}^{-L\otimes I_3 (t-\tau)}\right\| = \max_{\|\bold x\| = 1, t_0\leq \tau \leq t} \left\|\mathrm{e}^{-L\otimes I_3 (t-\tau)} \bold x \right\| \leq \eta \label{matrix_norm}
\end{eqnarray}
where $\eta>0$ is a finite value which bounds the matrix norm. It follows from (\ref{w_conv}) and (\ref{matrix_norm}) that
\begin{eqnarray}
\int^t_{t_0} \left\|\mathrm{e}^{-L\otimes I_3 (t-\tau)}\right\|  \left\|\bold w(\tau)  \right\| d\tau ~~~~~~~~~~~~~~~~~~~~\\
\leq \int^t_{t_0} \eta k_w \mathrm{e}^{-\lambda_w(\tau-t_0)}  \left\| \bold w(t_0) \right\| d\tau ~~~~~~~  \\
= \frac{-\eta k_w}{\lambda_w}  \left[\mathrm{e}^{-\lambda_w(\tau - t_0)} \right]^t_{t_0} \left\| \bold w(t_0) \right\| ~~~~~~~\\
=\frac{-\eta k_w}{\lambda_w} \left\| \bold w(t_0) \right\| (\mathrm{e}^{-\lambda_w(t - t_0)} - 1 ) \label{int_w}~~~~~~
\end{eqnarray}
From (\ref{int_w}), (\ref{err_bound}) can be written as follows:
\begin{eqnarray}
\| \bold e(t) - \bold e^*\| &\leq &  k_l \mathrm{e}^{-\lambda_l (t-t_0)} \| \bold e (t_0) - \bold e^*\| \\
&&-\frac{\eta k_w}{\lambda_w} \left\| \bold w(t_0) \right\| (\mathrm{e}^{-\lambda_w(t - t_0)} - 1 ). \label{err_bound1}
\end{eqnarray} 
By replacing $t_0$ with $t_0' = (t+t_0)/2$ in (\ref{err_bound1}), we obtain
\begin{eqnarray*}
\| \bold e(t) - \bold e^*\| &\leq &  k_l \mathrm{e}^{-\lambda_l \left(\frac{t-t_0}{2}\right)} \left\| \bold e (\frac{t+t_0}{2}) - \bold e^* \right\|\\
&& -\frac{\eta k_w}{\lambda_w} \left\| \bold w(\frac{t+t_0}{2}) \right\| (\mathrm{e}^{-\lambda_w\left(\frac{t-t_0}{2}\right)} - 1 ) \\
& \leq & k_l^2 \mathrm{e}^{-\lambda_l (t-t_0)} \| \bold e (t_0) - \bold e^*\| \\
&&-\frac{\eta k_w^2}{\lambda_w} \mathrm{e}^{-\lambda_w\left(\frac{t-t_0}{2}\right)} \left\| \bold w(t_0) \right\| (\mathrm{e}^{-\lambda_w\frac{t - t_0}{2}} - 1 )
\end{eqnarray*}
This completes the proof. 
\end{proof}
\section{Network Localization in 3-D Space}
Consider the agents under the assumption of \emph{Problem} \ref{prob_2}. We can obtain $\hat C_i$ from the proposed estimation method for the unknown orientation as stated in the previous sections. Then, we design an estimation law of $\hat p_i$ as follows:
\begin{eqnarray}\label{est_rule_p}
\dot{\hat p}_i = k_u \sum_{j\in \mathcal N_i} l_{ij}( \hat p_j - \hat p_i - \hat C_i^T p^i_{ji})
\end{eqnarray}
where $p_{ji}^i$ is a relative displacement defined in (\ref{eq_relative_pos}) and $k_u,l_{ij}>0$. We know that $p^i_{ji} = p^i_j- p^i_i = C_i(p_j - p_i)$. By defining $\tilde p_i$ as $\tilde p_i := \hat p_i - C^*p_i$, we obtain the estimation error dynamics as follows :
\begin{eqnarray}
\dot{\tilde{p}}_i &=& k_u \sum_{j\in \mathcal N_i} l_{ij}( \hat p_j - \hat p_i - \hat C_i^T C_i (p_{j} -p_i)) \nonumber \\
&=&k_u \sum_{j\in \mathcal N_i} l_{ij}\left( \tilde p_j - \tilde p_i\right) \nonumber \\
&& + \underbrace{k_u \sum_{j\in \mathcal N_i}l_{ij} (C^* - \hat C_i^T C_i )(p_{j} -p_i)}_{\psi_i :=} 
\end{eqnarray}
Thus the overall error dynamics for position estimation can be written as 
\begin{eqnarray}\label{est_err}
\dot{\tilde{\bold p}} = -k_u L\otimes I_3 \tilde{\bold{p}} + \Psi
\end{eqnarray}
where $\tilde{\bold p} = (\tilde p_1,\tilde p_2,\dots, \tilde p_N)$ and $\Psi = (\psi_1,\psi_2,\dots \psi_N)$. 
We notice that (\ref{est_err}) corresponds to ($\ref{cont_err}$). We obtain the following corollary based on the relationship. 
\begin{corollary}
Under the orientation estimation law (\ref{vec_z_model}) and the position estimation law (\ref{est_rule_p}), there exists a finite point $ \hat p^\infty_i, \forall i\in \mathcal V$ to which $\hat p_i $ globally exponentially converges, such that $\|\hat p^\infty_j-\hat p^\infty_i\|= \| p_j-p_i\|$, if $\mathcal G$ has a rooted-out branch, and $\bold{\hat q}_k(t_0), \forall k\in\{1,2\} $is not in $\rm{C}$$(L_H \otimes I_3)$. \label{thm_localization}
\end{corollary}
\section{Simulation Result}
In this section, we provide simulation results to verify the proposed methods. We consider six agents whose interaction graph is illustrated in Fig. \ref{fig_graph}.
\begin{figure}[!tb]
\begin{center}
\includegraphics[width=0.25\textwidth]{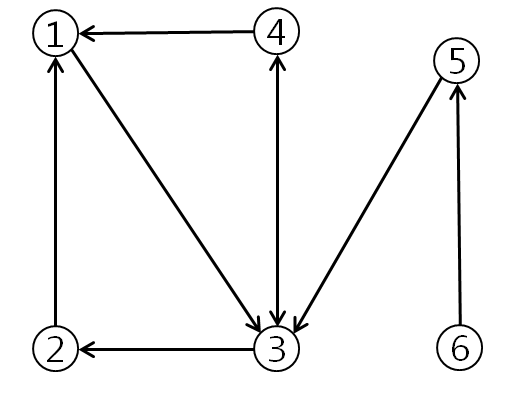}
\caption{Interaction graph of six agents for simulation}
\label{fig_graph}
\end{center}
\end{figure}
Suppose that each agent has its own reference frame which is rotated from the global reference frame by the proper orthogonal matrix $C_i\in \mathbb R^{3\times 3},i=\{1,2,3,4,5,6\}$. For the simulation, the matrix $C_i$ is arbitrary determined. Estimated orientation of $i$-th agent identified by $\hat C_i$ can be obtained by using two auxiliary variables. Under the estimation law (\ref{vec_z_model}), $\hat C_i(t)$ converges to a steady state solution as $t$ goes to infinity. The formation control law (\ref{cont_law}) of $i$-th agent is simultaneously conducted with the estimation of orientation $\hat C_i$. According to the result of the \emph{Theorem} \ref{thm_formation}, the position of $i$-th agent denoted by $p_i$ converges to $C^* p_i^* + p_\infty$, where $C^*\in \mathbb R^{3\times 3}$ is a common rotation matrix obtained by $ C^T \hat C_i(t)$ as $t\rightarrow \infty$, $p_i^*$ is desired position of $i$-th agent, and $p_\infty$ is a common position vector. Fig. \ref{fig_err} shows the result that the measured displacement of neighboring agents converges to the desired displacement with rotational motion. Fig. \ref{fig_formation} depicts that the formation of six agents using the proposed strategy, which includes a combination of orientation estimation and formation control law, converges to the rotated desired shape.

\begin{figure}[!tb]
\begin{center}
\includegraphics[width=0.5\textwidth]{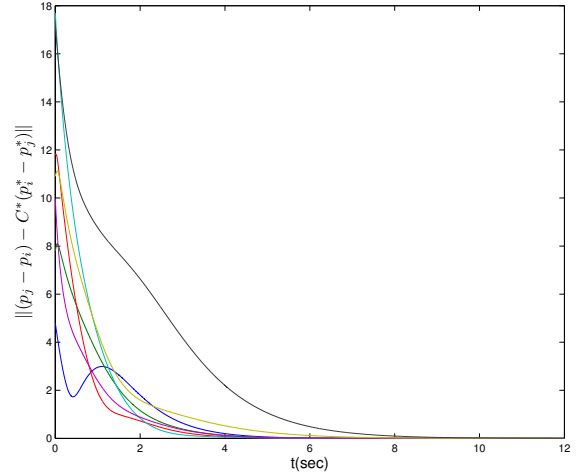}
\caption{Difference of the measured displacement and the desired displacement of neighboring agents. $C^*$ is a common rotation matrix.}
\label{fig_err}
\end{center}
\end{figure}

\begin{figure}[!tb]
\begin{center}
\includegraphics[width=0.55\textwidth]{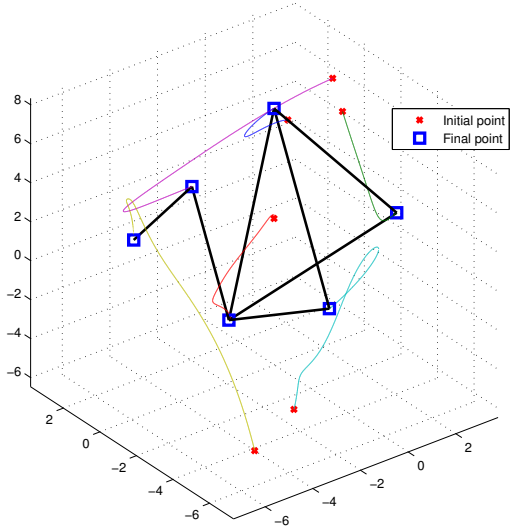}
\caption{Formation control via orientation estimation in 3 dimensional space}
\label{fig_formation}
\end{center}
\end{figure}

\section{Conclusion}
In this paper, we proposed a novel estimation method for unknown orientation of agent with respect to the common reference frame in general $n$-dimensional space. We added virtual auxiliary variables for each agent. These auxiliary variables are corresponding to column vectors of estimated orientation $\hat C_i 
\in \mathrm{SO}(n)$. An estimation law of auxiliary variables is based on the principle of consensus protocol.
 Since configuration space for the proposed method is $n$-dimensional vector space instead of nonlinear space like unit circle, global convergence of auxiliary variables is guaranteed. Consequently, each desired point to which auxiliary variables converge, under the proposed estimation law, is transformed to the estimated orientation. This estimated frame is used to compensate the rotated frame of each agent in the formation control system. Under the formation control law with proposed estimation method, we guarantee a global convergence of multi-agent system to the desired shape in 3-dimensional space, under the distance-based setup. 
 The proposed strategy can also be applied to the localization problem in networked systems. 
 The result remedies the weak point of previous works \cite{Oh:TAC2014,Montijano:ACC2014 } where the range of initial orientations is constrained to $0\leq \theta_i < \pi$. 

%


%
%

%
%

\end{document}